\documentclass[12pt]{article}

\usepackage{graphicx}    
\usepackage{amsmath} 
\usepackage{amsfonts} 
\usepackage{authblk}
\usepackage{xcolor}
\usepackage{amssymb, amsthm} 
\usepackage{booktabs}
\usepackage{hyperref}    
\usepackage{geometry}    
\usepackage{caption}     
\usepackage{cite}
\usepackage{algorithm}
\usepackage{algorithmicx}
\usepackage{algpseudocode}
\usepackage{pgfplots}
\usepackage{bm}
\usepackage{float}
\usetikzlibrary{patterns,positioning}
\usepackage{tikz}
\usetikzlibrary{decorations.markings}
\usetikzlibrary{matrix}
\usepackage{pgfplots}
\usetikzlibrary{shapes, arrows.meta, patterns}
\usepackage{mathrsfs}
\usepackage{circuitikz}
\usepackage{hhline}
\usepackage{multirow}
\usepackage{makecell} 

\usepackage{booktabs} 
 \usepackage{tabularx}

\usepackage{nicematrix, booktabs}

\usepackage{braket}

\usepackage{subcaption}

\DeclareMathAlphabet{\mathpzc}{OT1}{pzc}{m}{it}

\usepackage{etoolbox}

\newtheorem{theorem}{Theorem}
\newtheorem{definition}{Definition}

\usepackage{etoolbox} 
\makeatletter
\patchcmd{\blfootnote}{\@mpfootnotetext}{\@gobble}{}{}
\makeatother

\newtheoremstyle{mypropositionstyle} 
  {\topsep}                     
  {\topsep}                     
  {\itshape}                    
  {}                            
  {\bfseries}                   
  {.}                           
  {.5em}                        
  {\thmname{#1}\thmnumber{ #2}\thmnote{ (#3)}}

\theoremstyle{conjecturestyle}

\title{A Mechanistic Framework for \textit{\textbf{in Silico}} \textbf{Optimization} of Neuroblastoma Chemo-Immunotherapy}

\author[1, *]{Kate Brockman}
\author[3, *]{Brian Colburn}
\author[2, *]{Joseph Garza}
\author[3, *]{Yidong Liao}
\author[3, *]{B. Veena S.  N. Rao }

\affil[*]{All authors have contributed equally.}
\affil[1]{Department of Engineering, Texas A\&M University - Corpus Christi, TX- 78412, USA}
\affil[2]{Department of Life Sciences, Texas A\&M University - Corpus Christi, TX- 78412, USA}
\affil[3]{Department of Mathematics \& Statistics, Texas A\&M University-Corpus Christi, TX- 78412, USA}
\affil[ ]{\textit{E-mail address:} \texttt{bv.rao@tamucc.edu}}
\date{}

\begin{document}

\maketitle  
\begin{abstract}
A critical need exists for optimal therapeutic strategies for neuroblastoma, a prevalent and often fatal pediatric solid malignancy. To address the demand for quantitative models that can guide clinical decision-making, a novel mathematical framework was developed. Combination therapies involving immunotherapy, such as Interleukin-2 (IL-2), and chemotherapy, exemplified by Cyclophosphamide, have shown significant clinical potential by enhancing anti-tumor immune responses. In this study, a nonlinear system of coupled ordinary differential equations was formulated to mechanistically describe the interactions among tumor cells, natural killer (NK) cells, and cytotoxic T lymphocytes (CTLs). The pharmacodynamic effects of both IL-2 and Cyclophosphamide on these key immune populations were explicitly incorporated, allowing for the simulation of tumor dynamics across distinct patient risk profiles. The resulting computational framework provides a robust platform for the \textit{\textbf{in silico}} \textbf{optimization} of therapeutic regimens, presenting a quantitative pathway toward the improvement of clinical outcomes for patients with neuroblastoma.
\end{abstract}

\vspace{.1in}

 
\section{Introduction}
Neuroblastoma is a significant challenge in pediatric oncology, accounting for approximately 8\% of all childhood cancers and 15\% of related mortality~\cite{roy2017, smith2018}. Clinical outcomes are highly variable and depend on a standardized risk stratification system established by the International Neuroblastoma Risk Group (INRG). This system is crucial for optimizing treatment efficacy while minimizing long-term adverse effects, classifying patients based on factors like age, tumor stage, histology, and genetic markers (Table~\ref{tab:risk_categories}). Patients in the {low-risk} group are typically presented with localized tumors and have excellent prognoses, with 5-year survival rates of between 95--100\%; subsequent treatment is often limited to observation or surgical resection~\cite{franjic2024, anderson2022}. The {intermediate-risk} group has a survival rate exceeding 90\% and is generally treated with a combination of chemotherapy and surgery~\cite{mueller2009, franjic2024}. In contrast, {high-risk} patients face a more guarded prognosis, with 5-year survival rates between 45--50\%~\cite{franjic2024}. These patients receive intensive multimodal therapy that includes aggressive chemotherapy, surgery, radiation, and, increasingly, immunotherapy~\cite{arceci2009, mueller2009}. The integration of immunotherapy into these high-risk protocols underscores the critical need to better understand the complex interactions between the immune system and tumor cells to develop more effective therapeutic strategies for treating Neuroblastoma~\cite{anderson2022}.

The human anti-tumor immune response is orchestrated by two distinct yet coordinated systems: innate and adaptive immunity~\cite{getz2005}. The innate system, featuring Natural Killer (NK) cells, provides a rapid, non-specific first line of defense against malignant cells. In contrast, the adaptive system, mediated by cells such as {Cytotoxic T Lymphocytes (CTLs)}, develops a highly specific, antigen-dependent response that confers long-term immunological memory~\cite{loose2009}. This functional dichotomy is central to immuno-oncology, as the timing and nature of cell activation dictate therapeutic efficacy. Therefore, a comprehensive pharmacological model must incorporate both NK and CTL populations to accurately capture the complementary roles and temporal dynamics of the complete anti-tumor immune cascade~\cite{raulet2004}.

Therapeutic regimens for neuroblastoma are increasingly designed to leverage this dual immune architecture through combination strategies. {Cyclophosphamide}, a conventional cytotoxic chemotherapy agent, primarily acts by inducing apoptosis in rapidly dividing tumor cells, thereby reducing the overall tumor burden~\cite{menard2008}. Concurrently, {Interleukin-2 (IL-2)}, a potent cytokine, serves as an immunotherapeutic agent that stimulates enhanced cytotoxic activity of both NK cells and CTLs~\cite{berthold2000}. The strategic combination of these agents is designed to create a synergistic effect: chemotherapy reduces tumor load, while immunotherapy amplifies the host's innate and adaptive capacity to eliminate remaining malignant cells.

The complex, nonlinear dynamics arising from the interplay between tumor growth, dual-component immune response, and combination therapy necessitate the use of mathematical modeling~\cite{bellomo1997,benzekry2014,kuznetsov1994,padder2023,depillis2008,quaranta2005,song2018}. Mathematical models have been shown to be indispensable tools in oncology, providing quantitative frameworks which integrate complex biological processes and simulate treatment outcomes~\cite{quaranta2005}. Specifically in immuno-oncology, these models allow for the systematic investigation of therapeutic synergies, the optimization of dose-scheduling, and the formulation of patient-specific treatment strategies~\cite{katt2016}. By serving as \textit{in silico} platforms, they enable the exploration of testable hypotheses and the prediction of disease evolution under various therapeutic scenarios. While challenges related to biological complexity and clinical validation persist, the model developed herein aims to address these by providing a mechanistic framework to elucidate the principles governing neuroblastoma response to chemo-immunotherapy~\cite{quaranta2005, katt2016}.

In this study, a mechanistic mathematical model was developed to investigate the complex dynamics of neuroblastoma under combination chemo-immunotherapy. The model, formulated as a nonlinear system of coupled ordinary differential equations, explicitly describes the tripartite interactions among the neuroblastoma tumor cell population, Natural Killer (NK) cells of the innate immune system, and Cytotoxic T Lymphocytes (CTLs) of the adaptive immune system. Notably, a key feature within this framework lies in the integration of the pharmacodynamic effects of two distinct therapeutic agents: the cytotoxic chemotherapy drug Cyclophosphamide, which directly reduces tumor burden, and the immunomodulatory cytokine Interleukin-2 (IL-2), which stimulates the proliferation and activity of immune effector cells. This comprehensive model was then utilized as an in silico laboratory to simulate tumor progression and regression across clinically relevant patient subgroups, accounting for the different dynamics expected in low, intermediate, and high-risk disease profiles. Ultimately, this work provides a robust computational platform for systematically analyzing and optimizing treatment regimens, offering a quantitative approach to designing more effective, synergistic therapeutic strategies with the goal of improving clinical outcomes for children diagnosed with neuroblastoma.

\section{Foundations of tumor growth modeling}

The construction of a robust mathematical framework to model neuroblastoma requires a foundational understanding of the canonical frameworks that describe tumor proliferation. The selection of an appropriate growth function is a critical first step, as it establishes the dynamic core upon which the complex interactions with the immune system and therapeutic agents are built. The models reviewed below, frequently cited in oncological literature, represent a hierarchy of increasing biological realism and were systematically evaluated in the development of our framework~\cite{bellomo1997,benzekry2014,kuznetsov1994,padder2023,depillis2008,quaranta2005,song2018}.

\paragraph{Exponential Growth.} The most common and fundamental model used for modeling population dynamics is exponential growth, which assumes cells proliferate at a constant rate, $r$, in an environment with unlimited resources. This framework accurately describes the initial phase of tumor development, where spatial and nutrient constraints are negligible~\cite{benzekry2014}.
\begin{equation}
    \frac{dT}{dt} = rT \label{eq:exponential_growth}
\end{equation}

\paragraph{Logistic Growth.} In contrast to the unbounded nature of exponential growth, the logistic model introduces the concept of a carrying capacity, $K$, representing the maximum tumor size an environment can sustain. Here, the relative growth rate decreases linearly as the tumor population $T$ approaches $K$, effectively modeling the competition for limited space and nutrients that characterizes later-stage tumor progression~\cite{benzekry2014}.
\begin{equation}
    \frac{dT}{dt} = rT\left(1 - \frac{T}{K}\right) \label{eq:logistic_growth}
\end{equation}

\paragraph{Gompertz Growth.} The Gompertz model offers a refinement of density-dependent inhibition, featuring a characteristic exponential decay of the specific growth rate over time. Additionally, while also incorporating a carrying capacity, the Gompertzian curve is asymmetric, which captures that empirically observed phenomenon where solid tumor growth slows more rapidly than predicted by the logistic model. Its strong validation against experimental data has made it a standard in the field~\cite{benzekry2014}.
\begin{equation}
    \frac{dT}{dt} = rT \ln\left(\frac{K}{T}\right) \label{eq:gompertz_growth}
\end{equation}

\paragraph{Von Bertalanffy and Power Law Growth.} Further advancements are grounded in biophysical principles. The von Bertalanffy model posits that net growth arises from the balance between metabolic anabolic-based (synthesis) and catabolic (degradation) processes, which scale allometrically with tumor size~\cite{benzekry2014}. Similarly, the power-law model is often justified by geometric arguments, assuming that only a fraction of cells---typically those near the surface with access to vasculature---are actively proliferating, leading to a declining growth fraction as the tumor expands.
\begin{equation}
    \text{von Bertalanffy:} \quad \frac{dT}{dt} = p T^{m} - q T^{n} \qquad \text{Power Law:} \quad \frac{dT}{dt} = r T^{p} \label{eq:advanced_growth}
\end{equation}

The subsequent development of the presented neuroblastoma model drew upon the principles of these established frameworks. A customized model was engineered to simulate the unique dynamics of neuroblastoma under pharmacological intervention, refined through an iterative process of parameterization, stability analysis, and validation against known biological behaviors. This multi-faceted formulation process, which integrates the biological basis with the model architecture and its subsequent numerical analysis, is conceptually summarized in Figure~\ref{fig:model-formulation}.

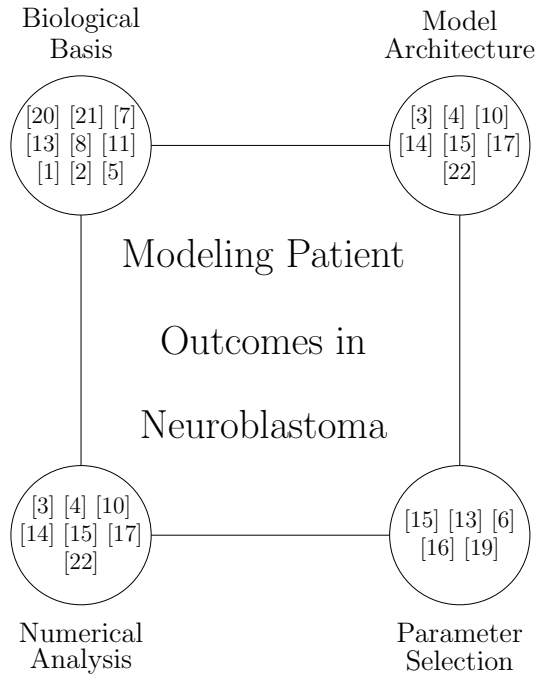
\begin{figure}[H]
    \centering
    \resizebox{0.5\textwidth}{!}{%
        \begin{circuitikz}
            \tikzstyle{every node}=[font=\Large]
            \draw  (3.5,10.25) circle (1.25cm);
            \node [font=\normalsize] at (3.5,10.75) {[20] [21] [7]};
            \node [font=\LARGE] at (12.75,4.25) {};
            \node [font=\normalsize] at (3.5,10.25) {[13] [8] [11]};
            \node [font=\normalsize] at (3.5,9.75) {[1] [2] [5]};
            \draw  (10.25,10.25) circle (1.25cm);
            \node [font=\normalsize] at (10.25,10.75) {[3] [4] [10]};
            \node [font=\normalsize] at (10.25,10.25) {[14] [15] [17]};
            \node [font=\normalsize] at (10.25,9.75) {[22]};
            \draw  (3.5,3.25) circle (1.25cm);
            \node [font=\normalsize] at (3.5,3.75) {[3] [4] [10]};
            \node [font=\normalsize] at (3.5,3.25) {[14] [15] [17]};
            \node [font=\normalsize] at (3.5,2.75) {[22]};
            \draw  (10.25,3.25) circle (1.25cm);
            \node [font=\normalsize] at (10.25,3.5) {[15] [13] [6]};
            \node [font=\normalsize] at (10.25,3) {[16] [19]};
            \node [font=\LARGE] at (6.75,8.25) {Modeling Patient};
            \node [font=\LARGE] at (6.75,6.75) {Outcomes in};
            \node [font=\LARGE] at (6.75,5.25) {Neuroblastoma};
            \draw [short] (3.5,9) -- (3.5,4.5);
            \draw [short] (4.75,3.25) -- (9,3.25);
            \draw [short] (10.25,4.5) -- (10.25,9);
            \draw [short] (9,10.25) -- (4.75,10.25);
            \node [font=\large] at (3.5,12.5) {Biological};
            \node [font=\large] at (3.5,12) {Basis};
            \node [font=\large] at (10.25,12.5) {Model};
            \node [font=\large] at (10.25,12) {Architecture};
            \node [font=\large] at (3.5,1.5) {Numerical};
            \node [font=\large] at (3.5,1) {Analysis};
            \node [font=\large] at (10.25,1.5) {Parameter};
            \node [font=\large] at (10.25,1) {Selection};
        \end{circuitikz}
    }
    \caption{A conceptual overview of the tumor-pharmacology model formulation process, integrating four key domains: the biological basis, the mathematical architecture, parameter selection, and numerical analysis.}
    \label{fig:model-formulation}
\end{figure}


\section{Mathematical modeling and analysis}

This section details the mathematical framework developed to investigate neuroblastoma dynamics under combination chemo-immunotherapy. We first present the core dimensional model, which builds upon established principles of tumor-immune interaction by incorporating key pharmacological agents. We then describe the nondimensionalization of the system, which simplifies the analysis and reveals the key parameters governing model behavior. Finally, we outline the numerical methods used for simulation.

\subsection{The chemo-immunotherapy model}

The foundation of our model is a system of coupled ordinary differential equations (ODEs) inspired by the tumor-immune interaction framework of Song et al.~\cite{song2022}. While their work captured baseline immune surveillance, our model introduces critical extensions by explicitly integrating the pharmacodynamics of {Interleukin-2 (IL-2)} immunotherapy and {Cyclophosphamide} chemotherapy, agents central to modern neuroblastoma treatment regimens. The resulting five-population system, given by Eq.~\eqref{eq:Tumor-pharmacology}, describes the time-evolution of Natural Killer (NK) cells ($N$), Cytotoxic T Lymphocytes (CTLs) ($L$), tumor cells ($T$), and the concentrations of IL-2 ($I$) and Cyclophosphamide ($C$).

\begin{equation}
\left\{
\begin{aligned}
N'(t) &= a_1 N(t)(1 - bN(t)) - a_2 N(t) - \alpha_1 N(t) T(t) + k_i I(t) \\
L'(t) &= r_1 N(t) T(t) - \mu L(t) - \beta_1 L(t) T(t) \\
T'(t) &= c T(t) (1 - d T(t)) - \alpha_2 N(t) T(t) - \beta_2 L(t) T(t) - k_c C(t) \\
I'(t) &= -\frac{\ln(2)}{h_i} I(t) \\
C'(t) &= -\frac{\ln(2)}{h_c} C(t)
\end{aligned}
\right.
\label{eq:Tumor-pharmacology}
\end{equation}

\paragraph{Cellular Dynamics.} The interactions between the aforementioned populations are defined as follows. The {NK cell population} ($N$) exhibits logistic growth (rate $a_1$, capacity $1/b$), natural decay (rate $a_2$), and depletion through interaction with tumor cells (rate $\alpha_1$). The {CTL population} ($L$) is activated in response to NK-tumor interactions (rate $r_1$), undergoes natural decay (rate $\mu$), and can become exhausted through sustained engagement with the tumor (rate $\beta_1$). The {tumor cell population} ($T$) follows a logistic growth pattern (rate $c$, capacity $1/d$) and is subject to killing by both NK cells (rate $\alpha_2$) and CTLs (rate $\beta_2$).

\paragraph{Pharmacological Dynamics.} The concentrations of IL-2 ($I$) and Cyclophosphamide ($C$) are modeled by standard first-order exponential decay, governed by their respective half-lives, $h_i$ and $h_c$. These pharmacokinetic profiles are coupled to the cellular dynamics to model their pharmacodynamic effects. IL-2 provides an external stimulus to the NK cell population (term $+k_i I(t)$), enhancing the innate immune response. Conversely, Cyclophosphamide exerts a direct cytotoxic effect on the tumor (term $-k_c C(t)$), representing the impact of chemotherapy. A comprehensive list of all model parameters, their descriptions, and units is provided in Table~\ref{tab:model_parameters}.

\begin{table}[H]
 	\centering
 	\captionsetup{justification=centering, singlelinecheck=false}
 	\caption{Tumor-pharmacology model parameters.}
 	\label{tab:model_parameters}
 	\small
 	\renewcommand{\arraystretch}{1.2}
 	\begin{tabular}{|l|p{8.5cm}|l|}
 		\hline
 		{Parameter} & {Description} & {Units} \\
 		\hline
 		$a_1$ & NK cell growth rate & $\text{cell} \cdot \text{day}^{-1}$ \\
 		$a_2$ & NK cell natural death rate & $\text{day}^{-1}$ \\
 		$b$ & Carrying capacity coefficient for NK cells & $\text{cell}^{-1}$ \\
 		$c$ & Tumor cell intrinsic growth rate & $\text{day}^{-1}$ \\
 		$d$ & Carrying capacity coefficient for tumor cells & $\text{cell}^{-1}$ \\
 		$\alpha_1$ & Rate of NK cell death from tumor interaction & $\text{cell}^{-1} \cdot \text{day}^{-1}$ \\
 		$\alpha_2$ & Rate of NK-induced tumor death (killing rate) & $\text{cell}^{-1} \cdot \text{day}^{-1}$ \\
 		$\beta_1$ & Rate of CTL death from tumor interaction & $\text{cell}^{-1} \cdot \text{day}^{-1}$ \\
 		$\beta_2$ & Rate of CTL-induced tumor death (killing rate) & $\text{cell}^{-1} \cdot \text{day}^{-1}$ \\
 		$\mu$ & CTL cell natural death rate & $\text{day}^{-1}$ \\
 		$r_1$ & Activation rate of CTLs from NK-tumor interaction & $\text{cell}^{-1} \cdot \text{day}^{-1}$ \\
 		$k_c$ & Efficacy of Cyclophosphamide-mediated tumor death & $\text{mg}^{-1} \cdot \text{day}^{-1} \cdot \text{cell}$ \\
 		$k_i$ & Efficacy of IL-2-mediated NK cell stimulation & $\text{IU}^{-1} \cdot \text{day}^{-1} \cdot \text{cell}$ \\
 		$h_i$ & Half-life of IL-2 & $\text{day}$ \\
 		$h_c$ & Half-life of Cyclophosphamide & $\text{day}$ \\
 		$I(0)$ & Initial dose of IL-2 & IU \\
 		$C(0)$ & Initial dose of Cyclophosphamide & mg \\
 		\hline
 	\end{tabular}
\end{table}

\subsection{Nondimensionalization and dynamic regimes}

To reduce the number of free parameters and elucidate the key mechanisms controlling the system's behavior, we nondimensionalized the model in Eq.~\eqref{eq:Tumor-pharmacology}. This process recasts the system in terms of dimensionless variables and composite parameters, facilitating a more general analysis of the underlying tumor-immune dynamics. The resulting nondimensional system is given by:
\begin{equation}
\left\{
\begin{aligned}
N'(t) &= p_1 N(t) (1 - q N(t)) - p_2 N(t) - N(t) T(t)  + k_i I(t)\\
L'(t) &= N(t) T(t) - L(t) - s L(t) T(t) \\
T'(t) &= u T(t) (1 - v T(t)) - N(t) T(t) - \delta L(t) T(t) - k_c C(t)
\end{aligned}
\right.
\label{eq:Nondimensional}
\end{equation}
where the drug concentration dynamics for $I(t)$ and $C(t)$ remain as in the original system. The thirteen dimensionless parameters are composites of the original seventeen, representing key biological ratios. For instance, $p_1 = a_1/\mu$ and $p_2 = a_2/\mu$ normalize the NK cell kinetics relative to the CTL decay rate, while $u=c/\mu$ scales the tumor growth rate. The parameter $\delta = \beta_2/\alpha_2$ represents the relative killing efficacy of CTLs compared to NK cells, and $s = \beta_1/\alpha_1$ is the ratio of CTL to NK cell exhaustion rates.

A theoretical analysis of this nondimensional framework reveals three characteristic dynamical regimes that correspond to distinct clinical outcomes:
\begin{itemize}
    \item {Tumor Elimination.} This optimal outcome occurs when the combined immune and therapeutic pressures overwhelm the tumor's proliferative capacity. The system converges to a tumor-free state, with the tumor population decaying towards zero.
    \item {Chronic Equilibrium.} A state of balance can be reached where the tumor is controlled but not eliminated by the immune system, resulting in a stable, non-zero tumor burden. This regime is often characterized by damped oscillations as the predator-prey-like dynamics between immune and tumor cells settle into equilibrium.
    \item {Tumor Escape.} This regime corresponds to treatment failure, where tumor proliferation outpaces the combined effects of immune surveillance and therapy. The tumor population grows uncontrollably, leading to unbounded disease progression.
\end{itemize}

\subsection{Numerical methods}

To explore the dynamic regimes associated with specific treatment scenarios, the system of ODEs was solved numerically. Simulations were implemented with \textsf{Python} using the \texttt{SciPy} library's integration routines. An adaptive-step solver was employed to ensure both accuracy and computational efficiency across the different timescales present in the model, from rapid drug decay to slower cellular population dynamics. This computational approach allowed for extensive exploration of the parameter space and a systematic comparison of treatment strategies across different patient risk profiles. A fractional-order version of the model was also used to incorporate memory effects into the immune-tumor interactions, with analysis focused on proving the existence and uniqueness of solutions.

This section outlines the computational and analytical methods used to investigate the neuroblastoma model. First, we describe the numerical simulations performed on the complete chemo-immunotherapy model to explore treatment outcomes. Second, we present a detailed theoretical analysis of a simplified, treatment-free version of the model using the framework of fractional calculus to prove that the underlying tumor-immune system is mathematically well-posed.

\subsection{Numerical simulation of the pharmacological model}
Numerical simulations of the full chemo-immunotherapy model (Eq.~\eqref{eq:Tumor-pharmacology}) were performed to investigate treatment efficacy across different clinical scenarios. Parameter values were adapted from existing literature (Table~\ref{tab:model_parameters}). At the same time, the initial conditions for tumor, NK, and CTL cell populations were set to reflect the characteristic cellular compositions of low, intermediate, and high-risk neuroblastoma. By running simulations for each risk category, we systematically assessed the model's dynamic behavior and predict treatment effectiveness under varied initial biological conditions.

\subsection{Theoretical analysis of a fractional-order tumor-immune model}
To capture the non-local, history-dependent nature of biological processes---often termed "memory effects"---we extended our analysis to a fractional-order version of the core tumor-immune interaction model. Standard integer-order derivatives are local in time, whereas fractional-order derivatives incorporate information from the entire history of a system. This analysis focuses on a simplified, treatment-free version of the model to rigorously establish its mathematical stability and coherence.

\paragraph{The Caputo fractional derivative.} We employ the Caputo fractional derivative, which is well-suited for initial value problems in physical systems. It is defined based on the Riemann-Liouville fractional integral.

\begin{definition}[Riemann-Liouville Integral]
For a function $\varphi(t)$ where $t>0$, the fractional integral of order $\alpha > 0$ is given by:
\begin{equation}
 I^\alpha \varphi(t) = \frac{1}{\Gamma(\alpha)} \int_{0}^{t} (t-\tau)^{\alpha-1} \varphi(\tau) \, d\tau
\end{equation}
where $\Gamma(\cdot)$ is the Gamma function. The kernel $(t-\tau)^{\alpha-1}$ weights the past values of the function $\varphi(\tau)$, giving the operator its memory property.
\end{definition}

\begin{definition}[Caputo derivative]
The Caputo fractional derivative of order $\alpha$ for a function $\varphi(t)$, where $n-1 < \alpha < n$, is defined as:
\begin{equation}
 D^\alpha\varphi(t) = I^{n-\alpha} D^{n}\varphi(t) = \frac{1}{\Gamma(n-\alpha)}\int_{0}^{t}(t-\tau)^{n-\alpha-1}\varphi^{(n)}(\tau)\,d\tau 
 \label{eq:caputo_def}
\end{equation}
where $D^n$ is the standard integer-order derivative.
\end{definition}

\paragraph{The fractional-order model formulation.} We begin with the non-dimensionalized, treatment-free subsystem describing the core interactions between NK cells ($x$), CTLs ($y$), and tumor cells ($z$). The integer-order system is:
\begin{equation}
\left\{
\begin{aligned}
x'(t) &= p_1 x(t) (1 - q x(t)) - p_2 x(t) - x(t) z(t) \\
y'(t) &= x(t) z(t) - y(t) - s y(t) z(t) \\
z'(t) &= u z(t) (1 - v z(t)) - x(t) z(t) - \delta y(t) z(t)
\end{aligned}
\right.
\end{equation}
By replacing the integer-order derivative with the Caputo operator $D^\alpha$ of order $0 < \alpha \leq 1$, we obtain the proposed fractional-order model:
\begin{equation}
\label{eq:fractional_system}
\left\{
\begin{aligned}
D^\alpha x(t) &= p_1 x(t) (1 - q x(t)) - p_2 x(t) - x(t) z(t) \\
D^\alpha y(t) &= x(t) z(t) - y(t) - s y(t) z(t) \\
D^\alpha z(t) &= u z(t) (1 - v z(t)) - x(t) z(t) - \delta y(t) z(t)
\end{aligned}
\right.
\end{equation}

\paragraph{Existence and uniqueness of the solution.} To ensure the model is well-posed, we must prove that a unique solution exists for a given set of initial conditions.

\begin{theorem}
For the system of fractional-order equations given in \eqref{eq:fractional_system}, with $0 < \alpha \leq 1$ and non-negative initial conditions, there exists a unique solution.
\end{theorem}

\begin{proof}
The proof relies on the Banach Fixed-Point Theorem. First, we rewrite the fractional differential system in its equivalent Volterra integral form. Let $Y(t) = [x(t), y(t), z(t)]^T$ be the state vector and $F(Y(t))$ be the vector of functions on the right-hand side of Eq.~\eqref{eq:fractional_system}. Applying the fractional integral operator $I^\alpha$ to both sides gives:
\begin{equation}
Y(t) = Y(0) + I^\alpha F(Y(t)) = Y(0) + \frac{1}{\Gamma(\alpha)} \int_0^t (t-\tau)^{\alpha-1} F(Y(\tau)) \,d\tau
\end{equation}
We define an operator $G$ on a continuous function space as:
\begin{equation}
G[Y(t)] = Y(0) + \frac{1}{\Gamma(\alpha)} \int_0^t (t-\tau)^{\alpha-1} F(Y(\tau)) \,d\tau
\end{equation}
The solution $Y(t)$ is a fixed point of the operator $G$. We prove that $G$ is a contraction mapping. The vector field $F$ is Lipschitz continuous, meaning there exists a constant $L > 0$ such that for any two state vectors $Y_1$ and $Y_2$, $\|F(Y_1) - F(Y_2)\| \leq L \|Y_1 - Y_2\|$.
Consider two solutions, $Y_1(t)$ and $Y_2(t)$, originating from the same initial condition. The difference is:
\begin{align*}
\| G[Y_1(t)] - G[Y_2(t)] \| &= \left\| \frac{1}{\Gamma(\alpha)} \int_0^t (t-\tau)^{\alpha-1} [F(Y_1(\tau)) - F(Y_2(\tau))] \,d\tau \right\| \\
&\leq \frac{1}{\Gamma(\alpha)} \int_0^t (t-\tau)^{\alpha-1} \|F(Y_1(\tau)) - F(Y_2(\tau))\| \,d\tau \\
&\leq \frac{L}{\Gamma(\alpha)} \int_0^t (t-\tau)^{\alpha-1} \|Y_1(\tau) - Y_2(\tau)\| \,d\tau
\end{align*}
Over a finite time interval $[0, T]$, this inequality leads to the condition that for a sufficiently small $T$, the operator $G$ is a contraction. By extending this argument over successive intervals, we can show that a unique solution exists for all $t \geq 0$. This guarantees that the model is deterministic and its behavior is uniquely dictated by its initial state.
\end{proof}

\paragraph{Stability of equilibrium points.} The long-term behavior of the system is determined by the stability of its equilibrium points.

\begin{theorem}
An equilibrium point $Y^*$ of the fractional-order system is locally asymptotically stable if all eigenvalues $\lambda$ of the Jacobian matrix $J(Y^*) = \frac{\partial F}{\partial Y}|_{Y=Y^*}$ satisfy the condition $|\arg(\lambda)| > \frac{\alpha\pi}{2}$.
\end{theorem}
This theorem provides the mathematical criteria to determine whether the system will converge to a state of tumor elimination, chronic equilibrium, or diverge in tumor escape, based on the model parameters evaluated at the system's fixed points. The detailed stability analysis is presented in Section~\ref{sec:stability analysis}.

\subsection{Stability analysis}
\label{sec:stability analysis}

To understand the long-term dynamics of the tumor-immune interaction, a stability analysis of the nondimensionalized system was performed. The equilibrium points (or steady states) of the system were identified, and the stability of each point was determined by analyzing the eigenvalues of the Jacobian matrix evaluated at that point. This method, known as linearization, reveals whether small perturbations from a steady state will decay (stability) or grow (instability). The general symbolic expressions for the non-trivial equilibrium points are algebraically complex; therefore, to gain clear biological insight, we focused our analysis on two key, interpretable steady states: the trivial (cell-free) state and the tumor-free (immune surveillance) state.

The Jacobian matrix, $J$, of the general nondimensionalized system (Eq.~\eqref{eq:Nondimensional}) is given by:
\[
J =
\begin{bmatrix}
p_{1}(1 - 2qN) - p_{2} - T & 0 & -N \\
T & -sT - 1 & N - sL \\
-T & -\delta T & u(1 - 2vT) - N - \delta L
\end{bmatrix}
\]

\paragraph{Trivial Equilibrium (Cell-Free State).}
The first equilibrium point is the trivial, cell-free state, $E_0 = (N^*, L^*, T^*) = (0, 0, 0)$. Evaluating the Jacobian at this point yields:
\[
J(E_0) =
\begin{bmatrix}
p_{1} - p_{2}  & 0 & 0 \\
0 & -1 & 0 \\
0 & 0 & u
\end{bmatrix}
\]
The eigenvalues are read directly from the diagonal: $\lambda_1 = p_1 - p_2$, $\lambda_2 = -1$, and $\lambda_3 = u$. Since the parameter $u$ (the scaled intrinsic tumor growth rate) is always positive, at least one eigenvalue is positive. {Therefore, the trivial equilibrium is always an unstable saddle point.} Biologically, this is expected: in an empty environment, the introduction of even a small number of tumor cells will lead to initial growth, causing the system to move away from the cell-free state.

\paragraph{Tumor-Free Equilibrium (Immune Surveillance State).}
The second key equilibrium represents a state where immune cells exist but the tumor has been eliminated: $E_1 = (N^*, L^*, T^*) = (1/q, 0, 0)$. This state exists provided $p_1 > p_2$. The Jacobian matrix evaluated at this point is:
\[
J(E_1) =
\begin{bmatrix}
-p_{1}  & 0 & -1/q \\
0 & -1 & 1/q \\
0 & 0 & u - 1/q
\end{bmatrix}
\]
This is an upper triangular matrix, so the eigenvalues are again the diagonal entries: $\lambda_1 = -p_1$, $\lambda_2 = -1$, and $\lambda_3 = u - 1/q$. Since $p_1$ is positive, the first two eigenvalues are always negative. The stability of this tumor-free state therefore depends entirely on the sign of $\lambda_3$.
\begin{itemize}
    \item {Case 1: Stable Immune Surveillance ($u < 1/q$).} If the scaled tumor growth rate $u$ is less than the NK cell carrying capacity $1/q$, then $\lambda_3$ is negative. All eigenvalues are negative, and the equilibrium $E_1$ is a {stable sink}. This condition represents a robust state of immune surveillance, where the NK cell population is sufficiently strong to eliminate any nascent tumor cells, keeping the system tumor-free.
    
    \item {Case 2: Immune Escape ($u > 1/q$).} If the tumor growth rate $u$ exceeds the NK cell carrying capacity $1/q$, then $\lambda_3$ becomes positive. With one positive eigenvalue, the equilibrium $E_1$ is an {unstable saddle point}. This condition represents a critical threshold for tumor escape. The immune system is no longer capable of controlling tumor growth, and any small perturbation will cause the tumor population to grow, driving the system towards a state of active disease.
\end{itemize}
This threshold behavior, where stability switches as a parameter crosses a critical value, is characteristic of a transcritical bifurcation and highlights the delicate balance between tumor proliferation and immune control. A summary of these findings is presented in Table~\ref{tab:eigenvalue_results}.

\begin{table}[H]
 	\centering
 	\captionsetup{justification=centering, singlelinecheck=false}
 	\caption{Stability conditions for the tumor-free equilibrium point $E_1$.}
 	\label{tab:eigenvalue_results}
 	\renewcommand{\arraystretch}{1.5}
 	\begin{tabular}{|l|l|p{0.4\textwidth}|} 
 		\hline
 		{Condition} & {Stability} & {Biological Interpretation} \\
 		\hline
 		$u < 1/q$ & Stable Sink & NK cell capacity exceeds tumor growth. {Immune Surveillance succeeds.} \\
 		\hline
 		$u > 1/q$ & Saddle Point (Unstable) & Tumor growth overwhelms NK cell capacity. {Tumor Escape occurs.} \\
 		\hline
 	\end{tabular}
\end{table}

\section{Results}
Simulations produced a total of 27 graphs illustrating the important aspects of immune cell dynamics and neuroblastoma tumor progression under various in silico treatments. These graphical visualizations revealed critical trends, including the temporal changes in tumor growth and immune cell populations, and demonstrated the relative efficacy of Cyclophosphamide and IL-2 with respect to control groups. Notably, the combined therapeutic strategy consistently outperformed single-agent applications, highlighting the potential benefit of multi-drug treatment approaches. 

\subsection{Parameter selection}
The initial conditions and parameters used in the model were informed by the interactions among tumor progression, drug concentration, and patient population, and were tailored to disease stage, with modifications adapted from the literature \cite{depillis2008}. Our approach to modeling the immunotherapeutic dynamics of neuroblastoma was grounded in the International Neuroblastoma Risk Group (INRG) staging system \cite{mueller2009}, which significantly enhanced our ability to compare patient populations in the context of therapeutic interventions.

\begin{table}[H]
    \centering
    \captionsetup{justification=raggedright, singlelinecheck=false}
    \makebox[\textwidth][l]{\hspace{1.5cm}\parbox{\textwidth}{
        \caption{Disease stages derived by the International Neuroblastoma Risk Group (INRG).}
        \label{tab:risk_categories}
    }}
    \small
    \renewcommand{\arraystretch}{1.45}
    \begin{tabular}{|l|p{10cm}|}
        \hline
        Disease Stage & Description \\
        \hline
        Low Risk & Patients with L1 (localized tumors in one area) or MS (asymptomatic with favorable biology and metastases limited to skin, liver, or bone marrow) are considered low risk. These patients typically require observation, with surgery or chemotherapy only if symptoms arise. \\
        Intermediate Risk & L2 (regional tumors with IDRFs) and MS with unfavorable biology (e.g., diploidy) are classified as intermediate risk. These tumors may need chemotherapy, with surgery recommended if possible. \\
        High Risk & M (distant metastases), MS with MYCN amplification, or L2 in patients over 18 months with unfavorable features are high risk. These patients require aggressive treatment, including chemotherapy, surgery, and stem cell therapy. \\
        \hline
    \end{tabular}
\end{table}

In this study, we employed a three-category system to stratify patient populations into low, intermediate, and high-risk groups based on the tumor stages described in Table \ref{tab:risk_categories}. This system supported the comparative analyses across trials and informed our initial conditions for the developed mathematical model. By aligning with an internationally accepted risk framework, we enhanced the clinical relevance and reproducibility of our simulations.

For each risk group, an estimated relative initial abundance of tumor cells, CTLs, and NK cells was made while paying particular attention to capturing the relative population sizes as shown in Table \ref{tab:initial_values}. These immune profiles enabled accurate simulations of treatment outcomes under varying immunological baselines, which allowed for the exploration of how immune systems can be influenced by various therapeutic responses. 

The low-risk patient population was characterized by a low tumor cell count and a highly active immune response, where NK cells offered immediate defense to tumor cells, whilst CTLs provided a targeted and long-term defense compared to NK cells. In the intermediate-risk patient-populations, initial tumor cell count was higher, showcasing an increasingly significant role of CTLs in a long-term immune response. Although NK cells served initially as the first line of defense, the increased tumor burden necessitated a coordinated immune response, with CTLs becoming increasingly critical for targeting and eliminating growing tumor cell populations. In high-risk patient populations, the tumor cell count was elevated, and the immune system faces greater challenges. While NK cells provided an early line of defense and gradually declined over time, CTLs proved to be essential for long-term tumor control. CTLs ability to recognize specific antigens and undergo clonal expansion enabled a sustained immune response against rapidly proliferating tumor cells.

\begin{table}[H]
    \centering
    \captionsetup{justification=raggedright, singlelinecheck=false}
    \makebox[\textwidth][l]{\hspace{3.2cm}\parbox{\textwidth}{
        \caption{Tumor-pharmacology parameter values.}
        \label{tab:model_parameter_values}
    }}
    \small
    \renewcommand{\arraystretch}{1.15}
    \begin{tabular}{|p{8.5cm}|l|}
        \hline
        Parameter & Value \\
        \hline
        NK cell growth rate ($a_1$) & 0.111 \\
        NK cell death rate due to natural death ($a_2$) & 0.0412 \\
        Carrying capacity coefficient for NK cell population ($b$) & 1.02e-09 \\
        Natural tumor cell growth rate ($c$) & 0.514 \\
        Carrying capacity coefficient for tumor cell population ($d$) & 1.02e-09 \\
        Rate of NK cell death due to tumor interaction ($\alpha_1$) & 1e-07 \\
        Rate of NK-induced tumor death ($\alpha_2$) & 3.23e-07 \\
        Rate of CTL-cell death due to tumor interaction ($\beta_1$) & 3.422e-10 \\
        Rate of CTL-induced tumor death ($\beta_2$) & 0.01245 \\
        CTL cell death rate due to natural death ($\mu$) & 0.02 \\
        Rate of NK-lysed tumor cell debris activation of CTLs ($r_1$) & 2.908e-11 \\
        Rate constant of Cyclophosphamide-mediated tumor death ($k_c$) & 0.9 \\
        Rate constant of IL-2-mediated stimulation ($k_i$) & 5e+04 \\
        Half-life of IL-2 ($h_i$) & 5 \\
        Half-life of Cyclophosphamide ($h_c$) & 537 \\
        Dose of IL-2 ($I_0$) & variable \\
        Dose of Cyclophosphamide ($C_0$) & variable \\
        \hline
    \end{tabular}
\end{table}

The model parameters listed in Table \ref{tab:model_parameters} were selected to reflect the biological processes driving neuroblastoma progression, immune system interactions, and the therapeutic effects of IL-2 and Cyclophosphamide. NK cell populations were regulated by intrinsic growth ($a_1$), death ($a_2$), and a carrying capacity constraint ($b$), while tumor-induced cytotoxicity was governed by the interaction rate ($\alpha_1$). CTL dynamics were shaped by activation through NK-lysed tumor byproducts ($r_1$), natural cell death ($\mu$), and tumor-induced depletion ($\beta_1$). Tumor proliferation was defined by the growth rate ($c$) and the carrying capacity ($d$), with immune-mediated reduction captured by NK and CTL-mediated lysis ($\alpha_2$, $\beta_2$).

\begin{table}[H]
    \centering
    \caption{Initial cell counts for low, intermediate, and high risk scenarios.}
    \label{tab:initial_values}
    \small
    \renewcommand{\arraystretch}{1.25}
    \begin{NiceTabular}{|c|c|c|c|}[hlines]
        \CodeBefore
            \rowcolors{1}{lightgray!25}{}[respect-blocks]
        \Body
            \toprule
            \Block{2-1}{Cell Type} & \Block{1-3}{Risk} & & \\
            & Low & Intermediate & High \\
            \midrule
            $N(0)$ & $10^5$ & $10^5$ & $10^4$ \\
            $L(0)$ & 20 & 20 & 20 \\
            $T(0)$ & $8 \times 10^5$ & $10^7$ & $10^7$ \\
            \bottomrule
    \end{NiceTabular}
\end{table}

Drug-related parameters included the rate of IL-2 stimulation of NK cells ($k_i$) and the rate of Cyclophosphamide induced tumor reduction ($k_c$), both of which play important roles in modulating immune response and tumor growth in neuroblastoma. Their concentrations were modeled using first-order exponential decay based on the respective drug half-lives ($h_i$, $h_c$), capturing the transient nature of these agents in systemic circulation. Initial dosages of IL-2 ($I_0$) and Cyclophosphamide ($C_0$) define the dosing schedule for treatment simulations shown in Table \ref{tab:dosing_schedule}.

\subsubsection{Pharmacological schedules}
Dosing schedules for Cyclophosphamide and Interleukin-2 (IL-2) were differentiated by treatment intensity (low-dose and high-dose) and were divided into initial and recurring administrations. For Cyclophosphamide, low-dose treatment involved an initial dose of 2.5~mg/kg followed by recurring doses of 2~mg/kg every 24 hours. High-dose treatment began with an initial dose of 30~mg/kg, with recurring doses of 25~mg/kg administered at 24-hour intervals \cite{emadi2009}. IL-2 was administered on a fixed 24-hour cycle. Low-dose schedules involved an initial and recurring dose of \(3 \times 10^6\) units, while high-dose schedules used \(6 \times 10^6\) units per dose on the same interval \cite{pressey2016}.

\begin{table}[H]
    \centering
    \caption{Dosing schedules for Cyclophosphamide and Interleukin-2.}
    \small
    \renewcommand{\arraystretch}{1.25}
    \begin{NiceTabular}{|c|c|c|c|}[hlines]
    \CodeBefore
        \rowcolors{1}{lightgray!20}{}[respect-blocks]
    \Body
        \toprule
        Drug & Dosage & Initial & Recurring \\
        \midrule
        \Block{2-1}{Cyclophosphamide} & Low & 2.5 mg/kg & 2 mg/kg \\
        & High & 30 mg/kg & 25 mg/kg \\
        \Block{2-1}{IL-2} & Low & -- & $3 \times 10^6$ units \\
        & High & -- & $6 \times 10^6$ units \\
        \bottomrule
    \end{NiceTabular}
    \label{tab:dosing_schedule}
\end{table}

\subsubsection{Therapeutic agent selection}
Two therapeutic agents were selected for modeling neuroblastoma treatment: Cyclophosphamide and IL-2. These drugs were chosen for their complementary mechanisms in which they affect the cellular population \cite{emadi2009,rossi1994}. Cyclophosphamide acts through direct tumor cell cytotoxicity, whilst IL-2 enhances immune system activation. By incorporating both agents into our model, we aimed to quantify how chemotherapy induced tumor cell death and immune-mediated tumor suppression independently and synergistically. This approach allowed for the exploration of direct cytotoxic effects and immune system activation across different stages of patient diseases.

Cyclophosphamide, a chemotherapy drug, directly targets tumor cells by slowing the rate at which they proliferate and inducing apoptosis. Chemically, the drug alkylates DNA, leading to the formation of cross-links that interfere with replication and trigger cell death, particularly in rapidly dividing cells. Beyond its direct cytotoxicity, Cyclophosphamide controls the tumor microenvironment by suppressing immune responses \cite{emadi2009}. Despite its effects, its potent ability to reduce tumor burden makes it a cornerstone in neuroblastoma treatment in modern medicine. 

IL-2, an immune-modulating cytokine, plays a critical role in the activation of NK cells, which are essential for facilitating the innate immune response against tumor cell development. In particular, IL-2 binds to receptors on NK cells, which promotes proliferation, activation, and cytotoxicity.  These activated NK cells then directly eliminate tumor cells by releasing perforin and granzymes, inducing apoptosis \cite{rossi1994}. Lastly, IL-2 enhances the expression of activating receptors in NK cells, improving their ability to recognize and target tumor cells, especially early on in the immune defense. In neuroblastoma treatment, this immune activation supports a shift toward a tumor-targeting immune environment \cite{rossi1994}.

\subsubsection{Baseline dynamics and model validation} 
The developed model was first tested using a set of parameters identified by the stability analysis 
that produced a stable, chronic equilibrium without treatment (Figure~\ref{fig:nondim-lowrisk-case1}). As shown, NK and CTL 
populations (N \& L) initially declined but began to recover by day 5, with NK cell resurgence 
supporting CTL stabilization. By day 15, both populations reached steady-state levels. Although 
the immune response alone was insufficient to eliminate the tumor population (T), it maintained 
the tumor burden below  $10^{4}$ cells. This validates the model’s ability to capture a state of chronic 
disease which provided a baseline for evaluating therapeutic interventions introduced in the 
following section

\begin{figure}[H]
    \centering
    \includegraphics[width=0.85\linewidth]{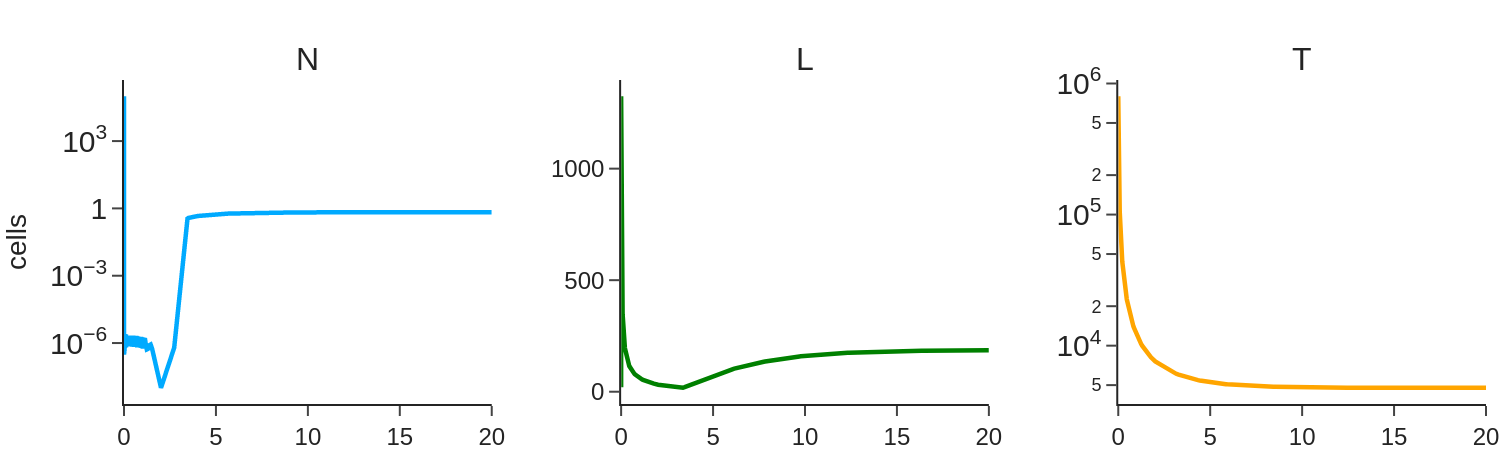}
    \captionsetup{justification=raggedright, singlelinecheck=false}
    \makebox[\textwidth][l]{\hspace{2.8cm}\parbox{\textwidth}{
        \caption{Numerical simulation results for stable parameters.}
        \label{fig:nondim-lowrisk-case1}
    }}
\end{figure}

\vspace{-1em}

\subsection{Simulation of treatment strategies}
Our analysis assessed various combinations of immunotherapy and chemotherapy across three distinct neuroblastoma patient populations stratified by risk level. Simulations were conducted to evaluate how treatment intensity and type influenced therapeutic outcomes within each group.

\section{Simulation methods and setup}

This section details the framework for the numerical simulations, including the basis for patient stratification, the selection of model parameters and initial conditions, and the design of the therapeutic regimens.

\subsection{Patient risk stratification and initial conditions}

Our simulation strategy was grounded in the clinically established International Neuroblastoma Risk Group (INRG) staging system, which classifies patients based on disease severity~\cite{mueller2009}. This approach ensured clinical relevance in our findings. The INRG stages, summarized in Table~\ref{tab:risk_categories}, were consolidated into three distinct risk profiles—low, intermediate, and high—to define the initial state of the system for simulations.

\begin{table}[H]
 	\centering
 	\captionsetup{justification=centering}
 	\caption{Disease stages as defined by the International Neuroblastoma Risk Group (INRG).}
 	\label{tab:risk_categories}
 	\begin{tabularx}{0.9\textwidth}{lX}
 		\toprule
 		{Disease Stage} & {Description} \\
 		\midrule
 		{Low Risk} & Localized tumors (L1) or metastatic disease with favorable biology (MS) confined to skin, liver, or bone marrow. Treatment is often limited to observation or surgery. \\
 		{Intermediate Risk} & Regional tumors with image-defined risk factors (IDRFs) (L2) or MS with unfavorable biology. These cases typically require chemotherapy in addition to surgery. \\
 		{High Risk} & Distant metastatic disease (M), MS with MYCN amplification, or advanced regional tumors in older patients. These cases require aggressive, multimodal therapy. \\
 		\bottomrule
 	\end{tabularx}
\end{table}

For each risk group, distinct initial cell populations for tumor cells ($T(0)$), NK cells ($N(0)$), and CTLs ($L(0)$) were established to reflect the corresponding biological state (Table~\ref{tab:initial_values}). The  \textit{low-risk} scenario was defined by a small tumor burden and a robust immune presence. The \textit{intermediate-risk} scenario featured a larger tumor and a moderately suppressed immune environment. The  \textit{high-risk} scenario was characterized by a very large tumor burden and a severely compromised immune cell count, representing an aggressive disease state where the immune system is overwhelmed.

\begin{table}[H]
 	\centering
 	\caption{Initial cell counts for the three simulated risk scenarios.}
 	\label{tab:initial_values}
 	\begin{tabular}{lccc}
 		\toprule
 		{Cell Type} & {Low Risk} & {Intermediate Risk} & {High Risk} \\
 		\midrule
 		NK cells, $N(0)$ & $10^5$ & $10^5$ & $10^4$ \\
 		CTLs, $L(0)$ & 20 & 20 & 20 \\
 		Tumor cells, $T(0)$ & $8 \times 10^5$ & $10^7$ & $10^7$ \\
 		\bottomrule
 	\end{tabular}
\end{table}

\subsection{Model parameters and therapeutic agents}

The model parameters, adapted from existing literature~\cite{depillis2008}, were chosen to reflect the underlying biological processes (Table~\ref{tab:model_parameter_values}). The therapeutic agents, Cyclophosphamide and Interleukin-2 (IL-2), were selected for their complementary mechanisms of action.

\paragraph{Cyclophosphamide.} A conventional cytotoxic chemotherapy agent, Cyclophosphamide primarily acts by alkylating DNA in rapidly dividing cells, which interferes with replication and induces apoptosis, thereby reducing the overall tumor burden~\cite{emadi2009}. Its potent ability to debulk tumors makes it a cornerstone of neuroblastoma treatment.

\paragraph{Interleukin-2 (IL-2).} An immune-stimulating cytokine, IL-2 plays a critical role in activating the innate immune response. It binds to receptors on NK cells, promoting their proliferation, activation, and cytotoxic function. Activated NK cells then eliminate tumor cells by releasing cytotoxic granules, making IL-2 a key agent for enhancing the anti-tumor immune environment~\cite{rossi1994}.

\begin{table}[H]
    \centering
    \caption{Parameter values used in the tumor-pharmacology model.}
    \label{tab:model_parameter_values}
    \begin{tabular}{llr}
        \toprule
        {Parameter} & {Description} & {Value} \\
        \midrule
        $a_1$ & NK cell growth rate & 0.111 \\
        $a_2$ & NK cell natural death rate & 0.0412 \\
        $b$ & NK cell carrying capacity coefficient & 1.02e-09 \\
        $c$ & Tumor cell intrinsic growth rate & 0.514 \\
        $d$ & Tumor cell carrying capacity coefficient & 1.02e-09 \\
        $\alpha_1$ & Rate of NK cell death from tumor interaction & 1e-07 \\
        $\alpha_2$ & Rate of NK-induced tumor death (killing rate) & 3.23e-07 \\
        $\beta_1$ & Rate of CTL death from tumor interaction & 3.422e-10 \\
        $\beta_2$ & Rate of CTL-induced tumor death (killing rate) & 0.01245 \\
        $\mu$ & CTL cell natural death rate & 0.02 \\
        $r_1$ & Activation rate of CTLs & 2.908e-11 \\
        $k_c$ & Efficacy of Cyclophosphamide & 0.9 \\
        $k_i$ & Efficacy of IL-2 stimulation & 5e+04 \\
        $h_i$ & Half-life of IL-2 (min) & 5 \\
        $h_c$ & Half-life of Cyclophosphamide (min) & 537 \\
        \bottomrule
    \end{tabular}
\end{table}

\subsection{Dosing schedules and simulation scenarios}
To investigate the impact of treatment intensity, we defined low-dose and high-dose schedules for both drugs, administered at 24-hour intervals (Table~\ref{tab:dosing_schedule}). These schedules were systematically applied to each of the three patient risk profiles, creating a matrix of 27 distinct simulation scenarios (Table~\ref{tab:simulation_scenario_matrix}) to comprehensively evaluate the interplay between disease severity and treatment strategy.

\begin{table}[H]
 	\centering
 	\caption{Dosing schedules for Cyclophosphamide and Interleukin-2.}
 	\label{tab:dosing_schedule}
 	\begin{tabular}{llcc}
 		\toprule
 		{Drug} & {Dosage Level} & {Initial Dose} & {Recurring Dose} \\
 		\midrule
 		\multirow{2}{*}{Cyclophosphamide} & Low & 2.5 mg/kg & 2 mg/kg \\
 		& High & 30 mg/kg & 25 mg/kg \\
 		\midrule
 		\multirow{2}{*}{IL-2} & Low & -- & $3 \times 10^6$ units \\
 		& High & -- & $6 \times 10^6$ units \\
 		\bottomrule
 	\end{tabular}
\end{table}

\begin{table}[H]
    \centering
    \captionsetup{justification=raggedright, singlelinecheck=false}
    \caption{Simulation matrix showing 9 applied treatments across three patient populations.}
    \label{tab:simulation_scenario_matrix_fixed}
    \renewcommand{\arraystretch}{1.2} 
    \small
    \begin{tabularx}{\textwidth}{@{} l l >{\raggedright\arraybackslash}X >{\raggedright\arraybackslash}X >{\raggedright\arraybackslash}X @{}}
        \toprule
        \multirow{2}{*}{\textbf{Risk}} & \multirow{2}{*}{\textbf{\makecell{Cyclophosphamide\\Dose}}} & \multicolumn{3}{c}{\textbf{IL-2 Dose}} \\
        \cmidrule(l){3-5} 
        & & \textbf{None} & \textbf{Low} & \textbf{High} \\
        \midrule
        \multirow{3}{*}{Low} & None & Low Risk, No Cyclophosphamide, No IL-2 & Low Risk, No Cyclophosphamide, Low IL-2 & Low Risk, No Cyclophosphamide, High IL-2 \\
        \cmidrule(l){2-5}
        & Low & Low Risk, Low Cyclophosphamide, No IL-2 & Low Risk, Low Cyclophosphamide, Low IL-2 & Low Risk, Low Cyclophosphamide, High IL-2 \\
        \cmidrule(l){2-5}
        & High & Low Risk, High Cyclophosphamide, No IL-2 & Low Risk, High Cyclophosphamide, Low IL-2 & Low Risk, High Cyclophosphamide, High IL-2 \\
        \midrule
        \multirow{3}{*}{Intermediate} & None & Intermediate Risk, No Cyclophosphamide, No IL-2 & Intermediate Risk, No Cyclophosphamide, Low IL-2 & Intermediate Risk, No Cyclophosphamide, High IL-2 \\
        \cmidrule(l){2-5}
        & Low & Intermediate Risk, Low Cyclophosphamide, No IL-2 & Intermediate Risk, Low Cyclophosphamide, Low IL-2 & Intermediate Risk, Low Cyclophosphamide, High IL-2 \\
        \cmidrule(l){2-5}
        & High & Intermediate Risk, High Cyclophosphamide, No IL-2 & Intermediate Risk, High Cyclophosphamide, Low IL-2 & Intermediate Risk, High Cyclophosphamide, High IL-2 \\
        \midrule
        \multirow{3}{*}{High} & None & High Risk, No Cyclophosphamide, No IL-2 & High Risk, No Cyclophosphamide, Low IL-2 & High Risk, No Cyclophosphamide, High IL-2 \\
        \cmidrule(l){2-5}
        & Low & High Risk, Low Cyclophosphamide, No IL-2 & High Risk, Low Cyclophosphamide, Low IL-2 & High Risk, Low Cyclophosphamide, High IL-2 \\
        \cmidrule(l){2-5}
        & High & High Risk, High Cyclophosphamide, No IL-2 & High Risk, High Cyclophosphamide, Low IL-2 & High Risk, High Cyclophosphamide, High IL-2 \\
        \bottomrule
    \end{tabularx}
\end{table}

\section{Results}

The 27 simulation scenarios were executed to quantify the therapeutic effects of Cyclophosphamide and IL-2, both as monotherapies and in combination, across the three neuroblastoma risk profiles. The results are presented below, with a focus on tumor burden over a 25-day period.

\subsection{Low-risk patient group}
In the low-risk scenario, characterized by a small initial tumor, the untreated control case resulted in slow tumor growth (Figure~\ref{fig:low-risk-results}). Low-dose Cyclophosphamide monotherapy provided only a temporary reduction in tumor burden before relapse occurred. In contrast, high-dose Cyclophosphamide or any regimen containing IL-2 proved sufficient to eradicate the tumor. The combination of high-dose chemotherapy and immunotherapy was so effective that it drove the tumor population to zero rapidly, demonstrating strong synergistic potential.

\begin{figure}[H]
    \centering
    \includegraphics[width=0.89\linewidth]{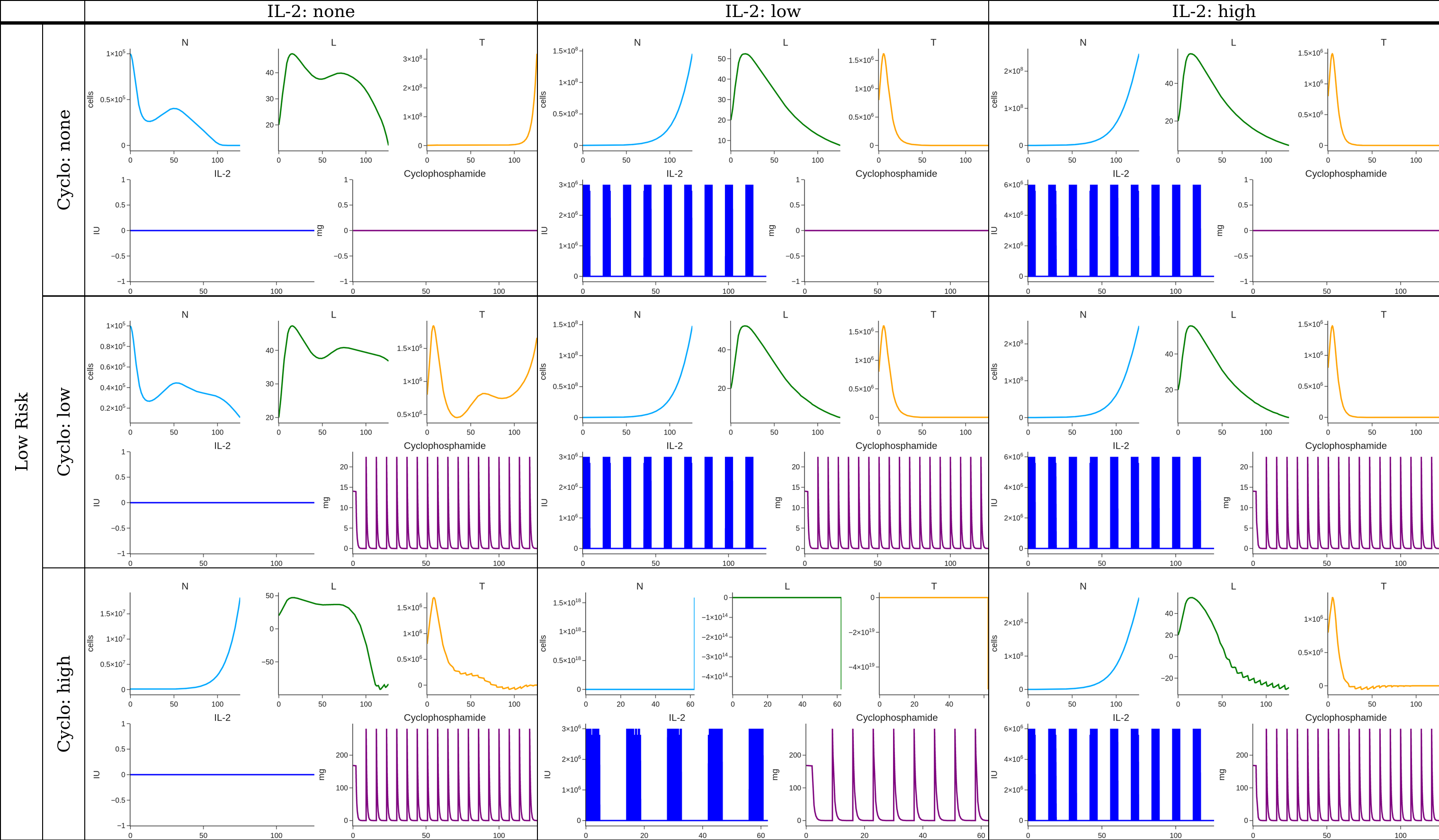}
    \caption{Simulation results for the low-risk patient group under nine treatment combinations.}
    \label{fig:low-risk-results}
\end{figure}

\subsection{Intermediate-risk patient group}
For the intermediate-risk cohort, with its higher initial tumor burden, monotherapy with either low- or high-dose Cyclophosphamide was insufficient to control tumor growth (Figure~\ref{fig:intermediate-risk-results}). The tumor's proliferative capacity quickly overwhelmed the cytotoxic effects of the chemotherapy. However, all treatment regimens that included IL-2 demonstrated a marked therapeutic effect, leading to tumor eradication. This highlights the critical role of immune stimulation in overcoming a more established tumor.

\begin{figure}[H]
    \centering
    \includegraphics[width=0.89\linewidth]{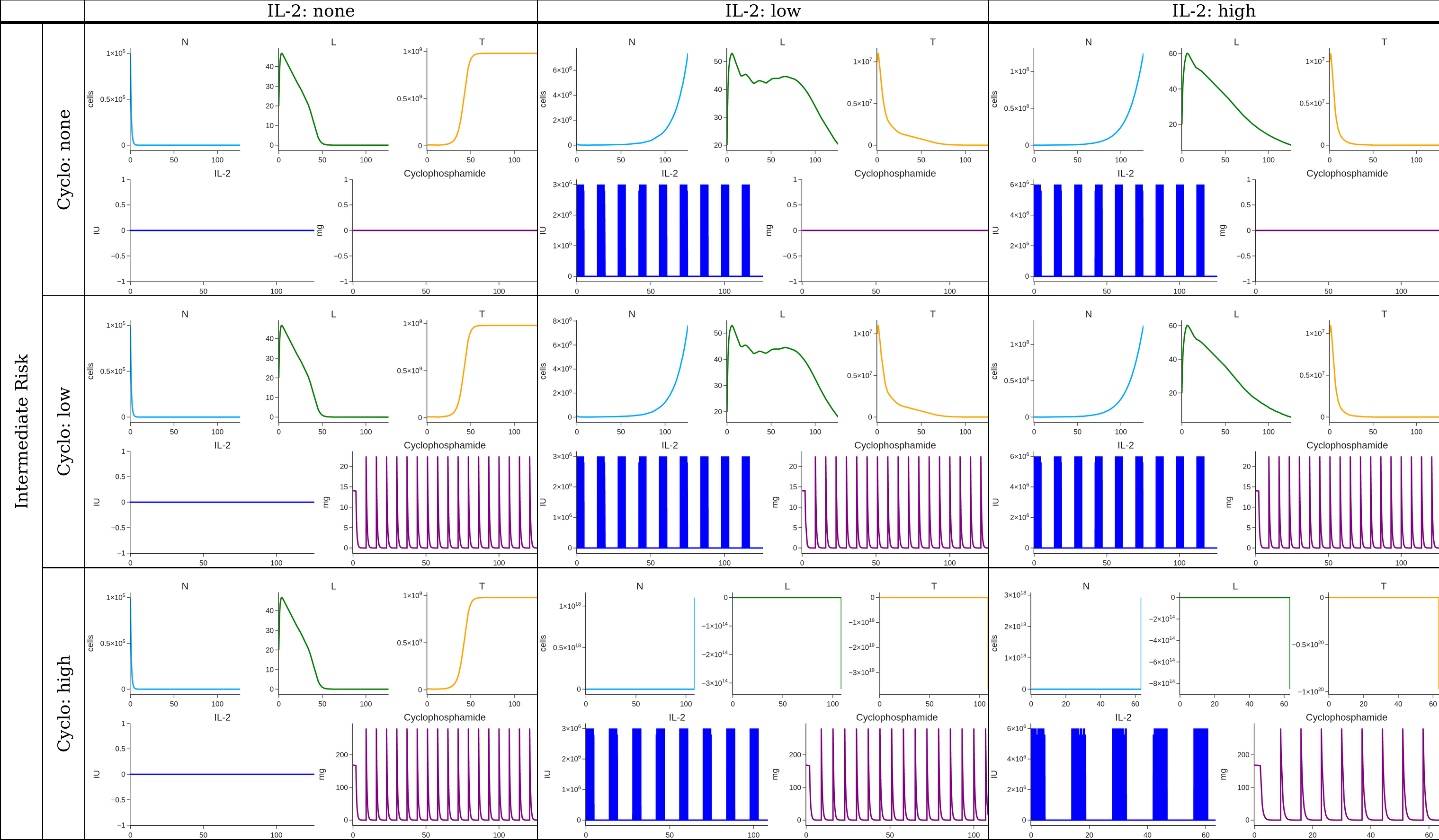}
    \caption{Simulation results for the intermediate-risk patient group.}
    \label{fig:intermediate-risk-results}
\end{figure}

\subsection{High-risk patient group}
The high-risk simulations modeled an aggressive disease state with a severely compromised immune system. In this challenging scenario, none of the tested therapeutic strategies were able to achieve tumor control (Figure~\ref{fig:high-risk-results}). While IL-2 administration produced transient bursts of immune activity, the immune cell populations were quickly exhausted and overwhelmed by the rapid proliferation of the large tumor. Cyclophosphamide monotherapy had no discernible impact on the tumor's growth trajectory. This result suggests that for high-risk disease, the tested regimens are insufficient, and alternative strategies such as higher doses, different scheduling, or additional therapeutic agents may be required.

\begin{figure}[H]
    \centering
    \includegraphics[width=0.89\linewidth]{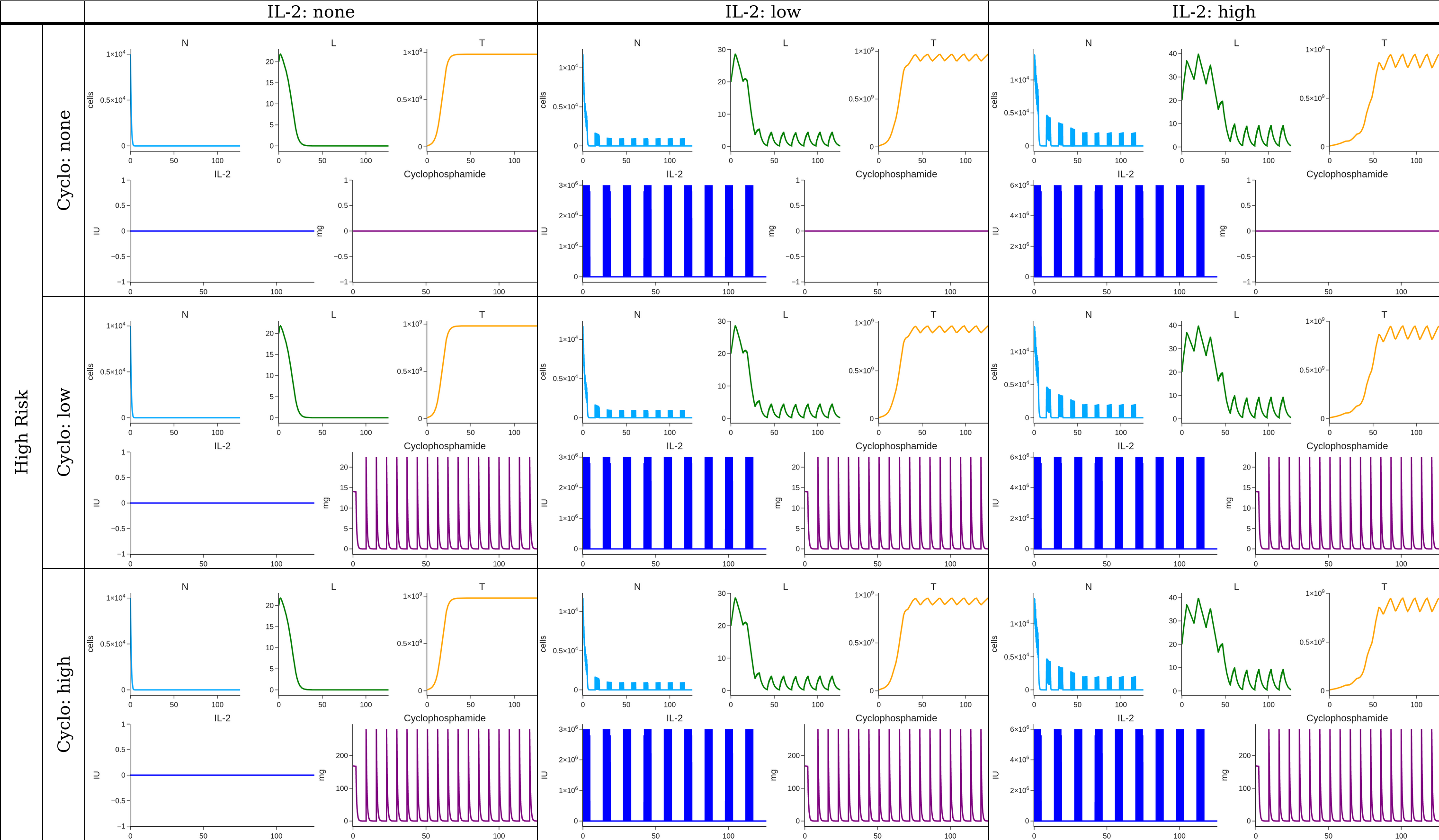}
    \caption{Simulation results for the high-risk patient group, showing treatment failure across all scenarios.}
    \label{fig:high-risk-results}
\end{figure}

\section{Conclusion}
In this study, the complex interplay between neuroblastoma cells, the innate and adaptive immune systems, and combination pharmacotherapy was investigated using a mechanistic mathematical model. Our framework, formulated as a system of coupled ordinary differential equations, successfully simulated the dynamics of tumor progression under various therapeutic scenarios tailored to clinically distinct patient risk profiles. The results underscore the critical role of both NK cells and CTLs in mediating anti-tumor immunity and quantitatively demonstrate that the strategic application of Interleukin-2 and Cyclophosphamide can synergistically enhance tumor elimination, particularly in low- and intermediate-risk conditions.

The findings have significant clinical implications. For low- and intermediate-risk neuroblastoma, our model suggests that immunotherapy, specifically IL-2, is a powerful component of treatment, capable of tipping the balance in favor of immune-mediated tumor control. However, the predicted failure of all tested regimens in the high-risk scenario is a crucial, albeit sobering, result. It suggests that the aggressive proliferation and profound immune suppression characteristic of high-risk disease may create a biological context in which these specific therapeutic agents, at conventional doses, are fundamentally insufficient. This highlights a critical need for novel therapeutic strategies---such as alternative cytokines, checkpoint inhibitors, or cellular therapies like CAR-T---to overcome the inherent resistance of advanced-stage neuroblastoma.

While this study provides valuable insights, it is important to acknowledge its limitations. The model is a deterministic simplification of a highly complex and stochastic biological system and does not account for spatial heterogeneity within the tumor microenvironment. Nonetheless, the presented framework serves as a robust and adaptable \textit{in silico} platform. Future work should focus on integrating patient-specific data to calibrate and personalize the model, exploring optimal dosing schedules and sequencing, and incorporating additional therapeutic agents to identify novel, more effective combination strategies. Ultimately, this research reinforces the power of mathematical modeling as a tool to dissect therapeutic mechanisms, generate testable hypotheses, and guide the rational design of next-generation treatments for pediatric cancer.

\section*{Acknowledgments}
The authors gratefully acknowledge the Texas A\&M High Performance Research Computing center for providing the advanced computing resources used to conduct this research.

\bibliographystyle{plain}
\bibliography{ref}

\end{document}